\documentclass[aps,reprint,groupedaddress]{revtex4-2}

\usepackage{graphicx}
\usepackage{amsmath,amssymb,bm,braket}
\usepackage{enumitem}
\usepackage{hyperref}

\usepackage{amsthm}
\newtheorem{lemma}{Lemma}
\newtheorem{theorem}{Theorem}
\newtheorem*{example}{Example}

\usepackage{mathtools}
\DeclarePairedDelimiterX{\infdivx}[2]{(}{)}{#1\;\delimsize\|\;#2}

\newcommand{\Tr}{\mathrm{Tr}}

\newcommand{\s}{s_\mathrm{TD}}

\newcommand{\CE}{canonical ensemble}

\newcommand{\indexCE}{\mathrm{c}}
\newcommand{\ind}[1]{
  \ensuremath{{#1}}}
\newcommand{\indind}[2]{
  \ensuremath{{#1,#2}}}

\usepackage{color}

\begin{document}

\title{Stationarity of quantum statistical ensembles at first-order phase transition points}
\author{Yasushi Yoneta}
\email{yasushi.yoneta@riken.jp}
\affiliation{Department of Basic Science,
The University of Tokyo, 3-8-1 Komaba, Meguro, Tokyo 153-8902, Japan}
\affiliation{RIKEN Center for Quantum Computing,
2-1 Hirosawa, Wako City, Saitama 351-0198, Japan}
\date{\today}
\begin{abstract}
We study the dynamics of quantum statistical ensembles at first-order phase transition points of finite macroscopic systems. First, we show that at the first-order phase transition point of systems with an order parameter that does not commute with the Hamiltonian, any quantum state with a non-zero value of the order parameter always evolves towards a macroscopically distinct state after a sufficiently long time. From this result, we argue that stationarity required for statistical ensembles should be interpreted as stationarity on a sufficiently long but finite time scale. Finally, we prove that the density matrix of the squeezed ensemble, a class of generalized statistical ensembles proposed as the only concrete method of constructing phase coexistence states applicable to general quantum systems, is locally stationary on time scales diverging in the thermodynamic limit. Our results support the validity of the squeezed ensemble from a dynamical point of view and open the door to non-equilibrium statistical physics at the first-order phase transition point.
\end{abstract}
\maketitle

\section{Introduction} \label{sec:introduction}

Understanding the dynamics of quantum many-body systems is one of the most intriguing research topics in condensed matter physics and quantum statistical mechanics. There, statistical ensembles are used as the initial states of the dynamics \cite{Kubo1991}. Therefore, the statistical ensemble must correctly provide not only static properties but also dynamic properties.

Stationarity is one of the most fundamental dynamic properties of equilibrium states \cite{Callen1985}. Accordingly, the density matrix given by the statistical ensemble should also be stationary. However, stationarity is not apparent when the equilibrium state is specified by a noncommutative set of additive observables since these observables are, in general, not conserved quantities. Nevertheless, away from the first-order phase transition point, one can employ the canonical ensemble,
which is strictly invariant under the time evolution.

However, the {\CE} cannot generally be applied to the equilibrium state at the first-order phase transition point \cite{Gross2001,Yoneta2019}. First-order phase transitions are characterized by a discontinuous change in the equilibrium state as a function of intensive parameters, such as the temperature and the magnetic field. Thus, at the first-order phase transition point, there exist several equilibrium states for a single value of a set of the intensive parameters: a single-phase state, which is the equilibrium state immediately before or after the phase transition point, and phase coexistence states, in which several phases coexist spatially in various proportions. On the other hand, the {\CE} is specified one-to-one by a set of the intensive parameters. As a result, at the first-order phase transition point, the {\CE} can give only a certain (single-phase states in many cases) or a statistical mixture of the various equilibrium states with the same value of intensive parameters. Consequently, it has long been pointed out that in finite systems, the equilibrium state at the first-order phase transition point can fluctuate macroscopically in time \cite{Goldenfeld1992,Lebowitz1999}.

Equilibrium states at the first-order phase transition point can be uniquely specified by a proper additive quantity, called the `order parameter'. Hence, at the first-order phase transition point, the statistical ensembles specified by additive observables including the order parameter have been employed, such as the microcanonical and restricted ensemble. However, such ensembles are ill-defined or ill-behaved when some of the additive observables do not commute with each other \cite{Yoneta2021}. Thus, it is difficult to theoretically give the general form of the quantum state corresponding to the equilibrium state at the first-order phase transition point.

Recently, this fundamental problem has been solved by extending the generalized ensemble, called the squeezed ensemble (SE), in such a way that it is applicable to equilibrium states specified by noncommutative additive observables \cite{Yoneta2021}. It was proved that the SE correctly gives the density matrix for any equilibrium state at the first-order phase transition point and the thermodynamic functions. Therefore, using the SE, one can thoroughly analyze microscopic structures and thermodynamic properties at first-order phase transition points of general quantum systems. However, the dynamic properties of the SE still need to be clarified.

In this paper, we study the dynamics of the SE at the first-order phase transition point of finite macroscopic systems. We consider the exact unitary dynamics rather than the dissipative dynamics \cite{Kindermann1980,Wilming2017}. This paper is organized as follows. In Section~\ref{sec:setup}, we describe the setup and notation. In Section~\ref{sec:decay}, we prove that at the first-order phase transition point with an order parameter that does not commute with the other additive observables specifying the equilibrium state, any quantum state with a non-zero value of the order parameter always changes macroscopically after a sufficiently long time. We then argue that stationarity of the equilibrium state should be interpreted as stationarity on a sufficiently long but finite time scale, at least in the first-order phase transition point. Finally, in Section~\ref{sec:stationarity}, we prove that the density matrix of the SE is locally stationary on time scales diverging in the thermodynamic limit. From these results, we discuss that the SE is consistent with thermodynamics.
 \section{Setup} \label{sec:setup}

We consider a quantum spin system with short-range interactions on the $\nu$-dimensional hypercubic lattice $\Lambda=[-n,+n]^\nu$ with $N=(2n+1)^\nu$ sites. We assume that the equilibrium state for each $N$ can be specified uniquely by a set of $m$ additive quantities $(X_{\ind{0}}=U, X_{\ind{1}},\cdots,X_{\ind{m-1}})$, where $U$ is the internal energy. Let $(\hat{X}_{\indind{N}{0}}=\hat{H}_N,\hat{X}_{\indind{N}{1}},\cdots,\hat{X}_{\indind{N}{m-1}})$ be a set of the corresponding additive observables, where $\hat{H}_N$ is the Hamiltonian. We assume that each additive observable $\hat{X}_{\indind{N}{i}}$ is expressed as
\begin{align}
  \hat{X}_{\indind{N}{i}} = \sum_{\substack{j\in\mathbb{Z}^\nu \text{\ s.t.}\\\gamma_j(I_{\ind{i}}) \subset \Lambda}} \gamma_j(\hat{o}_{\ind{i}}),
\end{align}
where $\gamma_j$ is the $j$-lattice translation for $j\in\mathbb{Z}$ and $\hat{o}_{\ind{i}}$ is an $N$-independent local observable with support $I_{\ind{i}}$.
For simplicity, we assume open boundary conditions, but the arguments in this paper easily extend to the case of general boundary conditions.

Our final purpose is to prove stationarity in the thermodynamic limit: $N\to\infty$ while $X_{\ind{i}}/N$ is fixed. Therefore, we introduce additive quantities per site
\begin{align}
  x_{\ind{i}} \equiv X_{\ind{i}}/N \qquad (i = 0, 1, \cdots, m-1),
\end{align}
and corresponding observable
\begin{align}
  \hat{x}_{\indind{N}{i}} \equiv \hat{X}_{\indind{N}{i}}/N \qquad (i = 0, 1, \cdots, m-1).
\end{align}
For simplicity of notation, sets of $m$ physical quantities are denoted by bold symbols like
\begin{align}
  \bm{x} &= (x_{\ind{0}},x_{\ind{1}},\cdots,x_{\ind{m-1}}).
\end{align}

Let $\s$ be the thermodynamic entropy per site. According to thermodynamics, $\s$ is a function of $\bm{x}$. The first derivative of $\s$ is called the `(entropic) intensive parameter' \cite{Callen1985}. Let $\Pi_{\ind{i}}$ denote the intensive parameter conjugate to $X_{\ind{i}}$, i.e.,
\begin{align}
  \Pi_{\ind{i}}(\bm{x}) \equiv \frac{\partial \s}{\partial x_{\ind{i}}}(\bm{x}).\label{eq:def_Pi}
\end{align}
Particularly, $\beta(\bm{x}) \equiv \Pi_{\ind{0}}(\bm{x})$ is the inverse temperature of the equilibrium state specified by $\bm{x}$.

Below we consider the time evolution of the density matrix corresponding to the equilibrium state. Let $\bm{\Pi}$ be a set of intensive parameters of the equilibrium state,
\begin{align}
  \bm{\Pi} = (\Pi_{\ind{0}},\Pi_{\ind{1}},\cdots,\Pi_{\ind{m-1}}).
\end{align}
According to thermodynamics, the equilibrium state is macroscopically stationary under external fields coupled to $X_{\ind{i}}$ of magnitude
\begin{align}
  f_{\ind{i}} = -\Pi_{\ind{i}}/\Pi_{\ind{0}}. \label{eq:mech_force}
\end{align}
On the other hand, according to quantum mechanics, such a dynamics is generated by the Hamiltonian including the interactions with external fields \footnote{$\hat{H}_N=\hat{X}_{\indind{N}{0}}$ is the Hamiltonian of the completely isolated system and does not include the interactions with external fields.}:
\begin{align}
  \hat{G}_N(\bm{\Pi})
  &\equiv \hat{H}_N - f_{\ind{1}} \hat{X}_{\indind{N}{1}} - \cdots - f_{\ind{m-1}} \hat{X}_{\indind{N}{m-1}} \nonumber\\
  &= \hat{H}_N + \frac{\Pi_{\ind{1}}}{\Pi_{\ind{0}}} \hat{X}_{\indind{N}{1}} + \cdots + \frac{\Pi_{\ind{m-1}}}{\Pi_{\ind{0}}} \hat{X}_{\indind{N}{m-1}}. \label{eq:G_N}
\end{align}
We rescale the time variable and the Hamiltonian as
\begin{align}
  \tau &\equiv t / \Pi_{\ind{0}},\\
  \hat{\Gamma}_N(\bm{\Pi}) &\equiv \Pi_{\ind{0}} \hat{G}_N = \sum_i \Pi_{\ind{i}} \hat{X}_{\indind{N}{i}}, \label{eq:Gamma_N}
\end{align}
and define the Heisenberg operator as
\begin{align}
  \hat{A}(\tau) &\equiv e^{+ i \hat{\Gamma}_N \tau} \hat{A} e^{- i \hat{\Gamma}_N \tau}.
\end{align}

Consider the case where the equilibrium state is {\em not} at the first-order phase transition point. In this case, the canonical Gibbs state is one of the quantum states that correspond to that equilibrium state:
\begin{align}
  \hat{\rho}_N^\indexCE \propto \exp \left[- \sum_i \Pi_{\ind{i}} \hat{X}_{\indind{N}{i}} \right].
\end{align}
Obviously, we can write $\hat{\rho}_N^\indexCE$ as
\begin{align}
  \hat{\rho}_N^\indexCE \propto \exp \left[- \hat{\Gamma}_N(\bm{\Pi}) \right]. \label{eq:Nc-Gamma}
\end{align}
Therefore, this density matrix is strictly invariant under the dynamics generated by $\hat{\Gamma}_N(\bm{\Pi})$.

However, in the case where the equilibrium state is at the first-order phase transition point, we cannot employ the {\CE}. For example, at the first-order phase transition point due to the spontaneous symmetry breaking, the {\CE} always gives a mixture of all ordered single-phase states with the same weight \cite{Schulman1980,Binder1981,Challa1986,Vollmayr1993} and cannot give a macroscopically definite state
\footnote{
It is sometimes possible to obtain the desired state at the first-order phase transition point using the {\CE} by imposing clever boundary conditions \cite{Dobrushin1973,vanBeijeren1975,Landau2014}. In this case, however, one would artificially add boundary terms to additive observables $\hat{\bm{X}}_N$ and construct the {\CE} from a set of these observables $\hat{\bm{X}}_N'$. On the other hand, since the generator of the time evolution is determined purely from quantum mechanics, independent of such artificial manipulations in statistical mechanics, a set of observables appearing in the definition of $\hat{\Gamma}_N$ are the same as the original set of additive observables $\hat{\bm{X}}_N (\neq \hat{\bm{X}}_N')$. As a result, Eq.~\eqref{eq:Nc-Gamma} no longer holds. Therefore, the canonical Gibbs state $\hat{\rho}_N^\indexCE$ is not strictly time invariant.
}.
In the next section, we study the time evolution of the quantum state corresponding to the equilibrium state at the first-order phase transition point.

\begin{example}[two-dimensional transverse-field Ising model]
Consider a two-dimensional system defined by the Hamiltonian
\begin{align}
  \hat{H}_N &= - J \sum_{\langle  i, j \rangle} \hat{\sigma}_i^z \hat{\sigma}_j^z - g \sum_{i} \hat{\sigma}_i^x.
  \label{eq:d2Ising_Hamiltonian}
\end{align}
Here, $\langle i,j \rangle$ denotes the nearest neighbors. We take the coupling $J$ to be ferromagnetic ($J>0$).
For $g$ smaller than the critical value $g_C \simeq 3.044$, this model exhibits spontaneous symmetry breaking
\cite{Elliott1971,Pfeuty1971,Nagai1987,Rieger1999,Blote2002,Nakamura2003,Jongh1998},
and at low temperature there exist several equilibrium states for a single $\beta$, which can be distinguished by the order parameter
\begin{align}
  \hat{X}_{\indind{N}{1}} &= \sum_{i} \hat{\sigma}_i^z,
  \label{eq:d2Ising_Magnetization}
\end{align}
i.e., the equilibrium state can be uniquely specified not by a value of the internal energy $U$ alone, but by a set of values of $U$ and $X_{\ind{1}}$.
In other word, this system exhibits a first-order phase transition where $X_{\ind{1}}$ changes discontinuously with respect to its conjugate intensive parameter $\Pi_{\ind{1}}$ (or longitudinal magnetic field $f_{\ind{1}} = -\Pi_{\ind{1}}/\Pi_{\ind{0}}$) at $\Pi_{\ind{1}}=0$ (or $f_{\ind{1}}=0$).
Note that the order parameter $\hat{X}_{\indind{N}{1}}$ do not commute with $\hat{H}_N$. Thus, it is not a conserved quantity. Therefore, at the first-order phase transition point, this system can fluctuate macroscopically in time between the equilibrium states with the same $\beta$ and different values of $\hat{X}_{\indind{N}{1}}$.
\end{example}
 \section{Macroscopic non-stationarity of ordered states} \label{sec:decay}

In this section, we prove under several conditions that at the first-order phase transition point with an order parameter that does not commute with the other additive observables, {\em any} quantum state with a non-zero value of the order parameter always changes {\em macroscopically} after a sufficiently long time. We then argue that, at the first-order phase transition point, it is impossible to require strict time invariance for statistical ensembles and that stationarity of the equilibrium state should be interpreted as stationarity on a sufficiently long but finite time scale.

To be concrete, we consider a system that exhibits spontaneous symmetry breaking and thus has an equilibrium state with a non-zero value of the order parameter. Suppose that the system is macroscopic but of finite size, as is the case for real thermodynamic systems. That is, we consider the case where the number of sites $N$ is fixed to a certain macroscopic but finite value. Choose a set of additive dynamical quantity per site $\bm{x}$ to be at the first-order phase transition point and consider the time evolution of the equilibrium state specified by that $\bm{x}$. Then, the argument $\bm{\Pi}$ of the generator of the time evolution $\hat{\Gamma}_N(\bm{\Pi})$ is taken to the set of intensive parameters at $\bm{x}$; $\bm{\Pi}=\bm{\Pi}(\bm{x})$.

Here, let us clarify the definitions of the symmetry and the order parameter of the system. When the unitary transformation $\hat{V}$ keeps $\hat{\Gamma}_N(\bm{\Pi})$ invariant, i.e.,
\begin{align}
  \hat{V}^\dagger \hat{\Gamma}_N(\bm{\Pi}) \hat{V} = \hat{\Gamma}_N(\bm{\Pi}),
\end{align}
then we call $\hat{V}$ the `symmetry transformation at $\bm{\Pi}$'. Suppose that the group $F$ formed by symmetric transformations at $\bm{\Pi}$ is a compact group.
Let $Z(F)$ be a center of $F$, i.e.,
\begin{align}
  Z(F) \equiv \left\{ \hat{U} \in F \middle| \hat{U}\hat{V}=\hat{V}\hat{U} \text{\ for all\ } \hat{V} \in F \right\}.
\end{align}
Suppose that the additive observable $\hat{X}_{\indind{N}{1}}$ is an order parameter
that measures the spontaneous breaking of the symmetry for an element of $Z(F)$, written as $\hat{R}$. That is, we consider the case where
\begin{enumerate}[label={(\roman*)}]
  \item \label{cond:order_parameter}
    There exists a symmetry transformation $\hat{R} \in Z(F)$ such that
    \begin{align}
      \hat{R}^\dagger \hat{X}_{\indind{N}{1}} \hat{R} = - \hat{X}_{\indind{N}{1}}.
      \label{eq:order_parameter}
    \end{align}
\end{enumerate}

\begin{example}[two-dimensional transverse-field Ising model]
To illustrate the settings described above more concretely, let us take an example of the transverse-field Ising model defined by Eqs.~\eqref{eq:d2Ising_Hamiltonian} and \eqref{eq:d2Ising_Magnetization}. As symmetries of this system at the first-order phase transition point $\bm{\Pi}=(\beta,0)$, the following are known:
\begin{itemize}
  \item $\mathbb{Z}_2$ symmetry with respect to the spin rotation by $\pi$ around the $x$-axis
    \begin{align}
      \hat{R} = \exp \left[ i \frac{\pi}{2} \sum_{i} \hat{\sigma}_i^x \right]
    \end{align}
  \item $C_{4v}$ symmetry of the square lattice
\end{itemize}
Let $F$ be the group formed by these symmetry transformations (i.e., $F \cong \mathbb{Z}_2 \times C_{4v}$), then $\hat{R} \in Z(F)$ and $\hat{R}^\dagger\hat{X}_{\indind{N}{1}}\hat{R}=-\hat{X}_{\indind{N}{1}}$. Therefore, it can be confirmed that condition~\ref{cond:order_parameter} is indeed fulfilled.
\end{example}

In addition, we assume that
there is no accidental degeneracy in the eigenstates of $\hat{\Gamma}_N$.
More precisely, we assume that
\begin{enumerate}[label={(\roman*)}]
  \setcounter{enumi}{1}
  \item For any pair of eigenstates $\ket{E_1}$ and $\ket{E_2}$
    which belong to the same eigenvalue of $\hat{\Gamma}_N$,
    there exists a symmetry transformation $\hat{V} \in F$
    such that $\braket{E_1|\hat{V}|E_2} \neq 0$.
    \label{cond:accidental_degeneracy}
\end{enumerate}

Under the above conditions~\ref{cond:order_parameter}-\ref{cond:accidental_degeneracy}, as proved in Appendix~\ref{sec:decay_proof}, we have the following theorem:
\begin{theorem} \label{theorem:decay}
For any initial state $\hat{\rho}$,
the long-time average of the expectation value of $\hat{X}_{\indind{N}{1}}$ under the time evolution generated by $\hat{\Gamma}_N(\bm{\Pi})$ is exactly zero:
\begin{align}
  \lim_{T\to\infty} \frac{1}{T} \int_0^T d\tau \ \Tr \left[ \hat{X}_{\indind{N}{1}}(\tau) \hat{\rho} \right] = 0.
\end{align}
\end{theorem}
From this theorem, it follows that at the first-order phase transition point, {\em any} quantum state corresponding to an equilibrium state in which the order parameter $X_{\ind{1}}$ is non-zero always changes to a macroscopically different state after an infinitely long time, whether it is a single-phase state or phase coexistence state.

Since this theorem holds for {\em any} initial state, this macroscopic non-stationarity cannot be avoided even by choosing an appropriate ensemble. On the other hand, it has been observed experimentally that the ordered state does exist stably even at the first-order phase transition point where the order parameter does not commute with the Hamiltonian \cite{Bacon1955,DeGennes1963}. However, the experiment is only conducted on a finite time scale, not on a truly infinite time scale. Therefore, to reconcile these two facts, it seems reasonable to interpret stationarity of the equilibrium state as stationarity on a sufficiently long but finite time scale, at least at the first-order phase transition point.

 \section{Stationarity of the density matrix of the squeezed ensemble} \label{sec:stationarity}

As discussed in Section~\ref{sec:introduction}, one cannot generally employ conventional ensembles at the first-order phase transition point when the order parameter does not commute with the other additive observables that specify the equilibrium state.
In Ref.~\cite{Yoneta2021}, the author has solved this fundamental problem by extending the generalized ensemble, called the squeezed ensemble (SE), in such a way that it is applicable to equilibrium states specified by noncommutative additive observables. We here investigate the time evolution of the density matrix of the SE and prove that it is stationary on the time scale that diverges in the thermodynamic limit.
In this section, we consider general first-order phase transitions, not only those due to spontaneous symmetry breaking.

We first review the definition of the SE and its properties that will be used in our analysis. Interested reader can find more information in Ref.~\cite{Yoneta2021}. Let $\eta(\bm{x})$ be a polynomial defined up to the order of the product, with real coefficients in $m$ noncommutative variables $\bm{x}$. Suppose that $\eta$ satisfies the following conditions:
\begin{enumerate}[label={(\Alph*)}]
  \item \label{cond:A}
    $\eta(\hat{\bm{x}}_N)$ is self-adjoint for all $N$.
  \item \label{cond:B}
    $\s(\bm{x})-\eta(\bm{x})$ has the unique maximum point
    $\bm{x}_\mathrm{max}^\eta$ in the thermodynamic state space
    and is strongly concave in a neighborhood of $\bm{x}_\mathrm{max}^\eta$.
\end{enumerate}
Using this $\eta$, the density matrix of the squeezed ensemble (SE) is defined as
\begin{align}
  \hat{\rho}_N^\eta
  &\equiv \frac{\displaystyle e^{-N\eta(\hat{\bm{x}}_N)}}{\displaystyle \Tr\left[e^{-N\eta(\hat{\bm{x}}_N)}\right]}. \label{eq:rho_N^eta}
\end{align}

As shown in Ref.~\cite{Yoneta2021}, the density matrix $\hat{\rho}_N^\eta$ correctly gives the equilibrium state specified by $\bm{x}=\bm{x}_\mathrm{max}^\eta$ even when $\hat{\bm{x}}_N$ do not commute with each other.
Thus, one can investigate microscopic structures of the equilibrium state at the first-order phase transition point of general quantum systems by choosing $\eta$ such that $\bm{x}_\mathrm{max}^\eta$ coincides with $\bm{x}$ in that equilibrium state.

One can also obtain thermodynamic functions easily from the SE. In particular, intensive parameters can be calculated using the following formula:
\begin{align}
  \Pi_{\ind{i}}^\eta
  &\equiv \lim_{N\to\infty} \frac{\partial\eta}{\partial x_{\ind{i}}}(\bm{x}_N^\eta)
  = \Pi_{\ind{i}}(\bm{x}^\eta),
  \label{eq:tdforce-formula}
\end{align}
where $x_{\indind{N}{i}}^\eta \equiv \Tr \left[ \hat{x}_{\indind{N}{i}} \hat{\rho}_N^\eta \right]$ and $\displaystyle x_{\ind{i}}^\eta \equiv \lim_{N\to\infty} x_{\indind{N}{i}}^\eta$. Since $\eta$ is a known function, one can obtain $\Pi_{\ind{i}}$ just by calculating $\bm{x}_N^\eta$.

Below in this section we investigate the time evolution of the density matrix of the SE. From the formula~\eqref{eq:tdforce-formula} for intensive parameters, in the equilibrium state described by the SE, the set of intensive parameters are
\begin{align}
  \bm{\Pi}^\eta = (\Pi_{\ind{0}}^\eta,\Pi_{\ind{1}}^\eta,\cdots,\Pi_{\ind{m-1}}^\eta).
\end{align}
Hence, according to the setup described in Section~\ref{sec:setup}, we consider the time evolution generated by $\hat{\Gamma}_N(\bm{\Pi}^\eta)$. Then, as proved in Appendix~\ref{sec:stationarity_proof}, the following holds which implies stationarity of the SE on time scales diverging in the thermodynamic limit:
\begin{theorem} \label{theorem:stationarity}
If there exists a positive constant $\alpha$ such that for every $i=0,1,\cdots,m-1$ the following conditions are fullfield:
\begin{align}
  \left( x_{\indind{N}{i}}^\eta - x_{\ind{i}}^\eta \right)^2
  &= O(N^{-2\alpha}), \label{cond:stationarity_finite-size-effect}\\
  \Tr \left[ \left( \hat{x}_{\indind{N}{i}} - x_{\indind{N}{i}}^\eta \right)^2 \hat{\rho}_N^\eta \right]
  &= O(N^{-2\alpha}). \label{cond:stationarity_variance}
\end{align}
Then for any $\epsilon>0$ and $U \subset \mathbb{Z}^\nu$, there exists $T=\Theta(N^{\min [\alpha/(\nu+1),1/\nu]})$ such that for any local observable $\hat{A}$ supported on $U$ and $|\tau| < T$ the relative time variation of the expectation value of $\hat{A}$ is less that $\epsilon$:
\begin{align}
  \left|
    \Tr \left[ \hat{A}(\tau) \hat{\rho}_N^\eta \right]
  - \Tr \left[ \hat{A}(0) \hat{\rho}_N^\eta \right]
  \right|
  < \epsilon \| \hat{A} \|.
  \label{eq:stationarity}
\end{align}
\end{theorem}

First, we would like to comment on the conditions of this theorem. Equation~\eqref{cond:stationarity_finite-size-effect} and \eqref{cond:stationarity_variance} is saying that the finite-size effects and variance of $\hat{x}_{\indind{N}{i}}$ in the SE decay algebraically in the thermodynamic limit, respectively. These conditions are confirmed numerically for a specific model with noncommutative $\hat{\bm{x}}_N$ (see Section~IX of Ref.~\cite{Yoneta2021}). Note that all conditions can be checked using only the statistical-mechanical quantities of the SE. Therefore, even for phase coexistence states of general quantum systems for which there is no concrete construction method other than SE, we can easily check stationarity of the density matrix of the SE.

Next, we discuss the consequences of this theorem. Consider performing a macroscopic experiment where the initial state is the density matrix $\hat{\rho}_N^\eta$ of the SE. Let $\epsilon$ be the maximum value of relative uncertainties of the measurement devices for local observables used in the experiment. Then it follows from Theorem~\ref{theorem:stationarity} that the time variations of the expectation values of the local observables are within the uncertainties of the measurements, on time scales diverging in the thermodynamic limit. Therefore, as long as one is interested in the local observables, one cannot observe any time variation. Since in thermodynamics we are only interested in local observables and additive observables, which can be expressed as a sum of local observables, stationarity in the above seance is sufficient for stationarity of the equilibrium state. Thus, we conclude that the SE is consistent with thermodynamics.
 \section{Conclusion}
In this paper, we have studied the time evolution of the quantum statistical ensemble at the first-order phase transition point. Stationarity is one of the most fundamental properties of equilibrium states. Therefore, the density matrix given by the statistical ensemble should also satisfy stationarity.
However, as we have shown, at the first-order phase transition point of systems with an order parameter that does not commute with the other additive observables specifying the equilibrium state, any quantum state with a non-zero value of order parameter always evolves towards a macroscopically distinct state after a sufficiently long time. From this result, we have argued that, at least at the first-order phase transition point, stationarity required for statistical ensembles should be interpreted as stationarity on a sufficiently long but finite time scale.
Then We have proved that the density matrix of the squeezed ensemble, a class of generalized ensembles proposed as the only concrete method of constructing phase coexistence states applicable to general quantum systems, satisfies stationarity in the above sense.
These results support the validity of the SE from a dynamical point of view and open the door to non-equilibrium statistical physics at the first-order phase transition point.

\begin{acknowledgments}
We thank Y. Chiba and A. Shimizu for discussions.
This work was supported by The Japan Society for the Promotion of Science, KAKENHI No.~20J22982.
\end{acknowledgments}

\onecolumngrid
\appendix
\section{Proof of Theorem~\ref{theorem:decay}} \label{sec:decay_proof}
\begin{proof}[Proof of Theorem~\ref{theorem:decay}]
Since $Z(F)$ is an Abelian group, its irreducible representations are all one-dimensional and can be expressed with real function $\theta_\mu$ as
\begin{align}
  Z(F) \ni \hat{U} \longmapsto e^{i\theta_\mu(\hat{U})} \in M(\mathbb{C}^1),
\end{align}
where $\mu$ labels the irreducible representation \cite{Serre1977}. Moreover, let $(\mathcal{H}_\nu,\sigma_\nu)$ be the irreducible representation of $F$. Here $\nu$ labels the irreducible representation.

Let $\mathcal{H}$ denote the Hilbert space associated to this quantum system. The following map $f$ is a representation on $\mathcal{H}$ of the direct product group $Z(F) \times F$:
\begin{align}
  f: Z(F) \times F \ni (\hat{U},\hat{V}) \longmapsto \hat{U} \hat{V} \in F.
\end{align}
In fact, for any $\hat{U}_1,\hat{U}_2 \in Z(F)$ and $\hat{V}_1,\hat{V}_2 \in F$, since $\hat{U}_2$ is in the center of $F$, it holds
\begin{align}
  f(\hat{U}_1,\hat{V}_1)f(\hat{U}_2,\hat{V}_2)
  = \hat{U}_1 \hat{V}_1 \hat{U}_2 \hat{V}_2
  = \hat{U}_1 \hat{U}_2 \hat{V}_1 \hat{V}_2
  = f(\hat{U}_1\hat{U}_2,\hat{V}_1\hat{V}_2).
\end{align}
That is, $f$ is indeed a group homomorphism.

First, we consider the irreducible decomposition of $(\mathcal{H},f)$. Any irreducible representation of the direct product group $Z(F) \times F$ can be written as a tensor product of irreducilble representations of $Z(F)$ and $F$ \cite{Serre1977}. Thus, in the present case, any irreducible representation of $Z(F) \times F$ can be expressed as
\begin{align}
  f_{\mu\nu}: Z(F) \times F \ni (\hat{U},\hat{V})
  \longmapsto e^{i\theta_\mu(\hat{U})} \sigma_\nu(\hat{V}) \in \mathcal{H}_\nu.
\end{align}
Therefore, we have the irreducible decomposition of $(\mathcal{H},f)$ as
\begin{align}
  \mathcal{H} &\cong \bigoplus_{\mu,\nu} {\mathcal{H}_\nu}^{\oplus n_{\mu\nu}} \cong \bigoplus_{\mu,\nu} \underbrace{\mathbb{C}^{n_{\mu\nu}} \otimes \mathcal{H}_\nu}_{\equiv \mathcal{K}_{\mu\nu}}, \label{eq:irr_H}\\
  f &\cong \bigoplus_{\mu,\nu} {f_{\mu\nu}}^{\oplus n_{\mu\nu}} \cong \bigoplus_{\mu,\nu} e^{i\theta_\mu} \hat{1}_{\mathbb{C}^{n_{\mu\nu}}} \otimes f_\nu. \label{eq:irr_f}
\end{align}
Here, $\cong$ denotes the unitary equivalence, $\oplus$ denotes the direct sum, and $n_{\mu\nu}$ denotes the number of copies of the irreducible representation $(\mathcal{H}_\nu,f_{\mu\nu})$ in $(\mathcal{H},f)$. In particular, we obtain
\begin{align}
  \hat{R}
  = f(\hat{R},\hat{1})
\cong \bigoplus_{\mu,\nu} e^{i\theta_\mu(\hat{R})} \hat{1}_{\mathbb{C}^{n_{\mu\nu}}} \otimes \hat{1}_{\mathcal{H}_\nu}. \label{eq:irr_R}
\end{align}
On the other hand, because for any $(\hat{U},\hat{V}) \in Z(F) \times F$ it holds $[\hat{\Gamma}_N,f(\hat{U},\hat{V})]=0$ by the definition of $F$, it follows from Schur's lemma that there exists $\hat{\gamma}_{\mu\nu} \in M(\mathbb{C}^{n_{\mu\nu}})$
such that
\begin{align}
  \hat{\Gamma}_N \cong \bigoplus_{\mu\nu} \hat{\gamma}_{\mu\nu} \otimes \hat{1}_{\mathcal{H}_{\nu}}. \label{eq:irr_Gamma}
\end{align}

Next, we diagonalize $\hat{\Gamma}_N$. Since $\hat{\Gamma}_N$ is self-adjoint, $\hat{\gamma}_{\mu\nu}$ is also self-adjoint. Then let $\{\ket{i}_{\mu\nu}\}_{i=1,\cdots,n_{\mu\nu}}$ be an orthonormal basis of $\mathbb{C}^{n_{\mu\nu}}$ with respect to which $\hat{\gamma}_{\mu\nu}$ is diagonal and $\gamma_{\mu\nu}^i$ be the (real) eigenvalue associated with $\ket{i}_{\mu\nu}$.
Furthermore, we take an arbitrary orthonormal basis of $\mathcal{H}_\nu$ as $\{\ket{a}_{\nu}\}_{a=1,\cdots,\dim\mathcal{H}_\nu}$.
Using $\ket{i}_{\mu\nu}$ and $\ket{a}_{\nu}$,
we set $\ket{E_{\mu\nu}^{ia}} (\in \mathcal{K}_{\mu\nu} \subset \mathcal{H})$ as
\begin{align}
  \ket{E_{\mu\nu}^{ia}} \equiv \ket{i}_{\mu\nu} \otimes \ket{a}_{\nu}. \label{eq:Emunui_def}
\end{align}
Then $\{\ket{E_{\mu\nu}^{ia}}\}_{\mu,\nu,i,a}$ is a complete orthonomal basis of $\mathcal{H}$.
In addition, using Eqs.~\eqref{eq:irr_R} and \eqref{eq:irr_Gamma}, we obtain
\begin{align}
  \hat{R} \ket{E_{\mu\nu}^{ia}} &= e^{i\theta_\mu(\hat{R})} \ket{E_{\mu\nu}^{ia}},\\
  \hat{\Gamma}_N \ket{E_{\mu\nu}^{ia}} &= \gamma_{\mu\nu}^i \ket{E_{\mu\nu}^{ia}}.
\end{align}

Then we examine the properties of $\gamma_{\mu\nu}^i$.
Suppose that $\gamma_{\mu\nu}^{i} = \gamma_{\mu'\nu'}^{i'}$.
Then, condition~\ref{cond:accidental_degeneracy} implies that
there exists $\hat{V} \in F$ such that
\begin{align}
  \braket{E_{\mu\nu}^{ia}|\hat{V}|E_{\mu'\nu'}^{i'a'}} \neq 0. \label{eq:Emunui_accidental_degeneracy}
\end{align}
On the other hand, using Eq.~\eqref{eq:irr_f} and \eqref{eq:Emunui_def}, we have
\begin{align}
  \ket{E_{\mu\nu}^{ia}} &\in \mathcal{K}_{\mu\nu}\\
  \hat{V} \ket{E_{\mu'\nu'}^{i'a'}} &\in \mathcal{K}_{\mu'\nu'}.
\end{align}
Since $\mathcal{K}_{\mu\nu} \cap \mathcal{K}_{\mu'\nu'} = \{0\}$ form Eq.~\eqref{eq:irr_H},
Eq.~\eqref{eq:Emunui_accidental_degeneracy} holds only when $\mu=\mu',\nu=\nu'$.
Furthermore, when $\mu=\mu',\nu=\nu'$,
since we have
\begin{align}
  \braket{E_{\mu\nu}^{ia}|\hat{V}|E_{\mu'\nu'}^{i'a'}}
  = \braket{i|i'}_{\mu\nu} \braket{a|\sigma_\nu(\hat{V})|a'}_{\nu},
\end{align}
Eq.~\eqref{eq:Emunui_accidental_degeneracy} holds only when $i=i'$. Thus, we get
\footnote{
Conversely, Eq.~\eqref{eq:accidental_degeneracy} implies condition~\ref{cond:accidental_degeneracy}. Suppose Eq.~\eqref{eq:accidental_degeneracy} holds. Then for any eigenstates $\ket{E_1}$ and $\ket{E_2}$ which belong to the same eigenvalue of $\hat{\Gamma}_N$, we can expand as
\begin{align}
  \ket{E_s}
  &= \sum_a c_s^a \ket{E_{\mu\nu}^{ia}}
  = \ket{i}_{\mu\nu} \otimes \left(\sum_a c_s^a \ket{a}_\nu\right)
  = \ket{i}_{\mu\nu} \otimes \ket{\phi_s}_\nu.
\end{align}
On the other hand, since $(\mathcal{H}_\nu,\sigma_\nu)$ is an irreducible representation of $F$, there exists $\hat{V} \in F$ such that $\braket{\phi_1|\sigma_\nu(\hat{V})|\phi_2} \neq 0$. Therefore, we have condition~\ref{cond:accidental_degeneracy}.
}
\begin{align}
  \gamma_{\mu\nu}^{i} = \gamma_{\mu'\nu'}^{i'}
  \Longleftrightarrow \mu=\mu',\nu=\nu',i=i'.
  \label{eq:accidental_degeneracy}
\end{align}

Finaly, we evaluate the long-time average of the expectation value of $\hat{X}_{\indind{N}{1}}(\tau)$. Due to the completeness of $\{\ket{E_{\mu\nu}^{ia}}\}_{\mu,\nu,i,a}$, we have
\begin{align}
  \hat{X}_{\indind{N}{1}}(\tau) = \sum_{\substack{\mu,\nu,i,a\\ \mu',\nu',i',a'}} e^{+i(\gamma_{\mu\nu}^i-\gamma_{\mu'\nu'}^{i'})\tau} \ket{E_{\mu\nu}^{ia}} \braket{E_{\mu\nu}^{ia}| \hat{X}_{\indind{N}{1}} |E_{\mu'\nu'}^{i'a'}} \bra{E_{\mu'\nu'}^{i'a'}}.
\end{align}
Therefore, form Eq.~\eqref{eq:accidental_degeneracy}, we have
\begin{align}
  \lim_{T\to\infty} \frac{1}{T} \int_0^T d\tau \hat{X}_{\indind{N}{1}}(\tau) = \sum_{\substack{\mu,\nu,i,a,a'}} \ket{E_{\mu\nu}^{ia}} \braket{E_{\mu\nu}^{ia}| \hat{X}_{\indind{N}{1}} |E_{\mu\nu}^{ia'}} \bra{E_{\mu\nu}^{ia'}}.
\end{align}
On the other hand, by condition~\ref{cond:order_parameter}, we have
\begin{align}
  \braket{E_{\mu\nu}^{ia}| \hat{X}_{\indind{N}{1}} |E_{\mu\nu}^{ia'}}
  &= - \braket{E_{\mu\nu}^{ia}| \hat{R}^\dagger \hat{X}_{\indind{N}{1}} \hat{R} |E_{\mu\nu}^{ia'}} \nonumber\\
  &= - \braket{E_{\mu\nu}^{ia}| e^{-i\theta_\mu(\hat{R})} \hat{X}_{\indind{N}{1}} e^{+i\theta_\mu(\hat{R})} |E_{\mu\nu}^{ia'}}
  = - \braket{E_{\mu\nu}^{ia}| \hat{X}_{\indind{N}{1}} |E_{\mu\nu}^{ia'}}.
\end{align}
Therefore, we obtain
\begin{align}
  \lim_{T\to\infty} \frac{1}{T} \int_0^T d\tau \hat{X}_{\indind{N}{1}}(\tau) = 0.
\end{align}
\end{proof}

\section{Proof of Theorem~\ref{theorem:stationarity}} \label{sec:stationarity_proof}
In our arguments, the following lemma plays a key role.
\begin{lemma}[Lieb-Robinson Bound \cite{Lieb1972,Nachtergaele2006,Nachtergaele2010}] \label{lemma:Lieb-Robinson}
There exist positive constants $C$, $\mu$ and $v$,
such that for any observables $\hat{A}$ and $\hat{B}$
with finite supports $U \subset \Lambda$ and $V \subset \Lambda$, respectively,
and for any $\tau \in \mathbb{R}$,
\begin{align}
  \left\|\left[\hat{A}(\tau), \hat{B}\right]\right\|
  &\leq C \left\|\hat{A}\right\| \left\|\hat{B}\right\|
  \min \left[\left|U\right|,\left|V\right|\right]
  e^{-\mu\left(d(U,V)-v|\tau|\right)}.
\end{align}
where
$d(x,y)$ is the distance
which is defined to be the shortest path length
that one needs to connect $x$ to $y$.
\end{lemma}
\begin{proof}
This follows from Eq.~(2.15) in Ref.~\cite{Nachtergaele2010}.
\end{proof}

\begin{proof}[Proof of Theorem~\ref{theorem:stationarity}]
First, we construct a local observable which approximates $\hat{A}(\tau)$. Take a real number $l$ such that $1 \ll l \ll L$. Here $L$ is a linear dimension of the system. Let $\Lambda^\mathrm{in} \equiv \left\{ j \in \Lambda \mid d(U,j)<l \right\}$, then $\left|\partial \Lambda^\mathrm{in}\right|=O\left(l^{\nu-1}\right)$. Since $\hat{\Gamma}_N$ is additive, there exist observables $\hat{\Gamma}_N^\mathrm{in}$, $\hat{\Gamma}_N^\mathrm{out}$ and $\hat{\Gamma}_N^\mathrm{int}$ such that
\begin{itemize}
  \item $\hat{\Gamma}_N = \hat{\Gamma}_N^\mathrm{in}+\hat{\Gamma}_N^\mathrm{out}+\hat{\Gamma}_N^\mathrm{int}$,
  \item $\mathrm{supp\ } \hat{\Gamma}_N^\mathrm{in} = \Lambda^\mathrm{in}$,
  \item $\mathrm{supp\ } \hat{\Gamma}_N^\mathrm{out} = \Lambda \setminus \Lambda^\mathrm{in}$,
  \item $\left\| \hat{\Gamma}_N^\mathrm{int} \right\|=O\left(l^{\nu-1}\right)$.
\end{itemize}
Now, we define an observable on $\Lambda^\mathrm{in}$ as
\begin{align}
  \tilde{A}(\tau)
  &\equiv e^{+ i (\hat{\Gamma}_N-\hat{\Gamma}_N^\mathrm{int}) \tau} \hat{A} e^{- i (\hat{\Gamma}_N-\hat{\Gamma}_N^\mathrm{int}) \tau}
  = e^{+ i \hat{\Gamma}_N^\mathrm{in} \tau} \hat{A} e^{- i \hat{\Gamma}_N^\mathrm{in} \tau}.
\end{align}
Applying Lemma~\ref{lemma:Lieb-Robinson}, we have
\begin{align}
  \left\|\hat{A}(\tau) - \tilde{A}(\tau)\right\| &= \left\|
    e^{+i \hat{\Gamma}_N \tau} \hat{A} e^{-i \hat{\Gamma}_N \tau}
    - e^{+i \left(\hat{\Gamma}_N-\hat{\Gamma}_N^\mathrm{int}\right) \tau} \hat{A} e^{-i \left(\hat{\Gamma}_N-\hat{\Gamma}_N^\mathrm{int}\right) \tau}
  \right\| \nonumber\\
  &= \left\|\int_{0}^{\tau} \frac{d}{d\sigma} \left(
    e^{+i \left(\hat{\Gamma}_N-\hat{\Gamma}_N^\mathrm{int}\right) (\tau-\sigma)}
    e^{+i \hat{\Gamma}_N \sigma}
    \hat{A}
    e^{-i \hat{\Gamma}_N \sigma}
    e^{-i \left(\hat{\Gamma}_N-\hat{\Gamma}_N^\mathrm{int}\right) (\tau-\sigma)}
  \right) d\sigma \right\| \nonumber\\
  &= \left\|\int_{0}^{\tau}
    e^{+i \left(\hat{\Gamma}_N-\hat{\Gamma}_N^\mathrm{int}\right) \sigma}
    \left[
      \hat{\Gamma}_N^\mathrm{int},
      e^{+i \hat{\Gamma}_N \sigma}
        \hat{A}
      e^{-i \hat{\Gamma}_N \sigma}
    \right]
    e^{-i \left(\hat{\Gamma}_N-\hat{\Gamma}_N^\mathrm{int}\right) \sigma}
  d\sigma \right\| \nonumber\\
  &\leq \int_{0}^{|\tau|} \left\|\left[ \hat{\Gamma}_N^\mathrm{int}, \hat{A}(\sigma) \right]\right\| d\sigma \nonumber\\
  &\leq \frac{C}{\mu v} \left\|\hat{\Gamma}_N^\mathrm{int}\right\| \left\|\hat{A}\right\| \left|U\right| e^{-\mu\left(l-v|\tau|\right)} \nonumber\\
  &= O\left( l^{\nu-1} e^{-\mu\left(l-v|\tau|\right)} \right).
  \label{eq:stationarity_hat-tilde}
\end{align}
Thus, we see that $\tilde{A}(\tau)$ is an observable on $\Lambda^\mathrm{in}$ which approximates $\hat{A}(\tau)$ with an error of $O\left( l^{\nu-1} e^{-\mu\left(l-v|\tau|\right)} \right)$.

Now we evaluate the time derivative of the expectation value of $\hat{A}(\tau)$ in the SE:
\begin{align}
  - i \frac{d}{d\tau} \Tr \left[ \hat{A}(\tau) \hat{\rho}_N^\eta \right]
  &= \Tr \left[ \left[ \hat{\Gamma}_N, \hat{A}(\tau) \right] \hat{\rho}_N^\eta \right].
\end{align}
Decomposing $\hat{A}(\tau)$ into its local term $\tilde{A}(\tau)$ and nonlocal term $\hat{A}(\tau)-\tilde{A}(\tau)$, we have
\begin{align}
  - i \frac{d}{d\tau} \Tr \left[ \hat{A}(\tau) \hat{\rho}_N^\eta \right]
  &= \underbrace{\Tr \left[ \left[ \hat{\Gamma}_N, \tilde{A}(\tau) \right] \hat{\rho}_N^\eta \right]}_{\equiv (\ast 1)}
  + \underbrace{\Tr \left[ \left[ \hat{\Gamma}_N, \hat{A}(\tau)-\tilde{A}(\tau) \right] \hat{\rho}_N^\eta \right]}_{\equiv (\ast 2)}. \label{eq:dAdt_partitioning}
\end{align}

First, we evaluate $(\ast 1)$ in Eq.~\eqref{eq:dAdt_partitioning}. Since $\eta$ is a polynomial, expanding in power series around $\bm{x}^\eta$, $\eta(\hat{\bm{x}}_N)$ can be written as a finite sum as
\begin{align}
  \eta(\hat{\bm{x}}_N)
  = \hat{\Gamma}_N(\bm{\Pi}^\eta)/N
  - \sum_{K(\geq 2),\{i_k\}} c_{i_1 \cdots i_K} \prod_{k=1}^{K} \left(\hat{x}_{\indind{N}{i_k}}-x_{\ind{i_k}}^\eta\right)
  + \mathrm{const.},
\end{align}
where we have used Eqs.~\eqref{eq:Gamma_N} and \eqref{eq:tdforce-formula}.
Substituting this into $(\ast 1)$, we have
\begin{align}
  (\ast 1)
  &= \sum_{K,\{i_k\}} c_{i_1 \cdots i_K} N \Tr \left[ \left[ \prod_{k=1}^{K} \left(\hat{x}_{\indind{N}{i_k}}-x_{\ind{i_k}}^\eta\right), \tilde{A}(\tau) \right] \hat{\rho}_N^\eta \right]\\
  &= \sum_{K,\{i_k\}} c_{i_1 \cdots i_K} \underbrace{\sum_{k'=1}^K \Tr \left[
    \prod_{k=1}^{k'-1} \left(\hat{x}_{\indind{N}{i_k}}-x_{\ind{i_k}}^\eta\right)
    \hat{q}_{\indind{N}{i_{k'}}} 
    \prod_{k=k'+1}^K \left(\hat{x}_{\indind{N}{i_k}}-x_{\ind{i_k}}^\eta\right)
  \hat{\rho}_N^\eta \right]}_{\equiv (\ast\ast)}.
\end{align}
Here we set
\begin{align}
  \hat{q}_{\indind{N}{i}}
  &\equiv \left[ \hat{X}_{\indind{N}{i}}, \tilde{A}(\tau) \right]
  = \sum_{\substack{j\in\mathbb{Z}^\nu \text{\ s.t.}\\\Lambda^\mathrm{in} \ \cap \ \mathrm{supp\ } \gamma_j(\hat{o}_{\ind{i}}) \neq \emptyset}} \left[ \gamma_j(\hat{o}_{\ind{i}}), \tilde{A}(\tau) \right],
\end{align}
where in the last equality we have used the fact that $\tilde{A}(\tau)$ is supported on $\Lambda^\mathrm{in}$. Then, by the triangle inequality and the unitary invariance of the operator norm, we have
\begin{align}
  \left\| \hat{q}_{\indind{N}{i}} \right\|
  \leq \sum_{\substack{j\in\mathbb{Z}^\nu \text{\ s.t.}\\\Lambda^\mathrm{in} \ \cap \ \mathrm{supp\ } \gamma_j(\hat{o}_{\ind{i}}) \neq \emptyset}} 2 \left\| \hat{o}_{\ind{i}} \right\| \left\| \hat{A} \right\|
  = O(l^\nu). \label{eq:q_norm}
\end{align}
Since $K \geq 2$, applying the Cauchy-Schwartz inequality
\begin{align}
  \left| \Tr \left[ \hat{A} \hat{B}^\dagger\hat{\rho}_N^\eta \right] \right|^2
  \leq \Tr \left[ \hat{A} \hat{A}^\dagger\hat{\rho}_N^\eta \right] \Tr \left[ \hat{B} \hat{B}^\dagger\hat{\rho}_N^\eta \right],
\end{align}
we have
\begin{align}
  |(\ast\ast)|
  &\leq \left| \Tr \left[
    \hat{q}_{\indind{N}{i_1}}
    \prod_{k=2}^K \left(\hat{x}_{\indind{N}{i_k}}-x_{\ind{i_k}}^\eta\right)
  \hat{\rho}_N^\eta \right] \right| \nonumber\\
  &+ \left| \Tr \left[
    \left(\hat{x}_{\indind{N}{i_1}}-x_{\ind{i_1}}^\eta\right)
    \hat{q}_{\indind{N}{i_2}}
    \prod_{k=3}^K \left(\hat{x}_{\indind{N}{i_k}}-x_{\ind{i_k}}^\eta\right)
  \hat{\rho}_N^\eta \right] \right| \nonumber\\
  &+ \cdots \nonumber\\
  &+ \left| \Tr \left[
    \prod_{k=1}^{K-1} \left(\hat{x}_{\indind{N}{i_k}}-x_{\ind{i_k}}^\eta\right)
    \hat{q}_{\indind{N}{i_K}}
  \hat{\rho}_N^\eta \right] \right| \nonumber\\
  &\leq \sqrt{\Tr \left[
    \hat{q}_{\indind{N}{i_1}}
    \prod_{k=2}^{K-1} \left(\hat{x}_{\indind{N}{i_k}}-x_{\ind{i_k}}^\eta\right)
    \prod_{k=K-1}^{2} \left(\hat{x}_{\indind{N}{i_k}}-x_{\ind{i_k}}^\eta\right)
    \hat{q}_{\indind{N}{i_1}}^\dagger
  \hat{\rho}_N^\eta \right]
  \Tr \left[
    \left(\hat{x}_{\indind{N}{i_K}}-x_{\ind{i_K}}^\eta\right)^2
  \hat{\rho}_N^\eta \right]} \nonumber\\
  &+ \sqrt{\Tr \left[
    \left(\hat{x}_{\indind{N}{i_1}}-x_{\ind{i_1}}^\eta\right)^2
  \hat{\rho}_N^\eta \right]
  \Tr \left[
    \prod_{k=K}^{3} \left(\hat{x}_{\indind{N}{i_k}}-x_{\ind{i_k}}^\eta\right)
    \hat{q}_{\indind{N}{i_2}}^\dagger
    \hat{q}_{\indind{N}{i_2}}
    \prod_{k=3}^{K} \left(\hat{x}_{\indind{N}{i_k}}-x_{\ind{i_k}}^\eta\right)
  \hat{\rho}_N^\eta \right]} \nonumber\\
  &+ \cdots \nonumber\\
  &+ \sqrt{\Tr \left[
    \left(\hat{x}_{\indind{N}{i_1}}-x_{\ind{i_1}}^\eta\right)^2
  \hat{\rho}_N^\eta \right]
  \Tr \left[
    \hat{q}_{\indind{N}{i_K}}^\dagger
    \prod_{k=K-1}^{2} \left(\hat{x}_{\indind{N}{i_k}}-x_{\ind{i_k}}^\eta\right)
    \prod_{k=2}^{K-1} \left(\hat{x}_{\indind{N}{i_k}}-x_{\ind{i_k}}^\eta\right)
    \hat{q}_{\indind{N}{i_K}}
  \hat{\rho}_N^\eta \right]}.
  \label{eq:astast}
\end{align}
Furthermore,
\begin{align}
  \Tr \left[
    \left(\hat{x}_{\indind{N}{i}}-x_{\ind{i}}^\eta\right)^2
  \hat{\rho}_N^\eta \right]
  = \Tr \left[
    \left(\hat{x}_{\indind{N}{i}}-x_{\indind{N}{i}}^\eta\right)^2
  \hat{\rho}_N^\eta \right]
  + \left(x_{\indind{N}{i}}^\eta-x_{\ind{i}}^\eta\right)^2.
\end{align}
Therefore, under the conditions of the theorem (Eqs.~\eqref{cond:stationarity_finite-size-effect}-\eqref{cond:stationarity_variance}), we obtain
\begin{align}
  |(\ast1)| &= O\left(N^{-\alpha} l^\nu\right).
\end{align}

Next, we evaluate $(\ast 2)$ in Eq.~\eqref{eq:dAdt_partitioning}. Using H\"{o}lder's inequality and Eq.~\eqref{eq:stationarity_hat-tilde}, we obtain
\begin{align}
  |(\ast2)|
  &\leq \left\| \left[ \hat{\Gamma}_N, \hat{A}(\tau)-\tilde{A}(\tau)\right] \hat{\rho}_N^\eta \right\|_1 \nonumber\\
  &\leq 2 \left\| \hat{\rho}_N^\eta \right\|_1 \left\| \hat{\Gamma}_N \right\| \left\| \hat{A}(\tau)-\tilde{A}(\tau) \right\|
  = O\left(N l^{\nu-1} e^{-\mu\left(l-v|\tau|\right)}\right).
\end{align}

From the above, the time derivative of the expectation value in the SE can be evaluated as
\begin{align}
  \left| \frac{d}{d\tau} \Tr \left[ \hat{A}(\tau) \hat{\rho}_N^\eta \right] \right|
  &= O\left(N^{-\alpha} l^\nu\right) + O\left(N l^{\nu-1} e^{-\mu\left(l-v|\tau|\right)}\right).
\end{align}
Using the triangle inequality for integrals, we have
\begin{align}
  \left| \Tr \left[ \hat{A}(\tau) \hat{\rho}_N^\eta \right]
  - \Tr \left[ \hat{A}(0) \hat{\rho}_N^\eta \right] \right|
  &= O\left(|\tau| N^{-\alpha} l^\nu\right) + O\left(|\tau| N l^{\nu-1} e^{-\mu\left(l-v|\tau|\right)}\right).
\end{align}
Then, letting $l = v|\tau| + \frac{2}{\mu} \log N$,
we obtain
\begin{align}
  \left| \Tr \left[ \hat{A}(\tau) \hat{\rho}_N^\eta \right]
  - \Tr \left[ \hat{A}(0) \hat{\rho}_N^\eta \right] \right|
  = \begin{cases}
    \displaystyle O \left( |\tau| N^{-\alpha} (\log N)^\nu\right) & (v|\tau| \lesssim \log L)\\ \\
    \displaystyle O \left( |\tau|^{\nu+1} N^{-\alpha} \right) & (\log L \ll v|\tau| \ll L)
  \end{cases}.
\end{align}
Therefore,
we find that Eq.~\eqref{eq:stationarity} holds for $|\tau|<T=\Theta(N^{\min [\alpha/(\nu+1),1/\nu]})$.
\end{proof}
 \twocolumngrid

\bibliography{document}

\begin{thebibliography}{32}%
\makeatletter
\providecommand \@ifxundefined [1]{%
 \@ifx{#1\undefined}
}%
\providecommand \@ifnum [1]{%
 \ifnum #1\expandafter \@firstoftwo
 \else \expandafter \@secondoftwo
 \fi
}%
\providecommand \@ifx [1]{%
 \ifx #1\expandafter \@firstoftwo
 \else \expandafter \@secondoftwo
 \fi
}%
\providecommand \natexlab [1]{#1}%
\providecommand \enquote  [1]{``#1''}%
\providecommand \bibnamefont  [1]{#1}%
\providecommand \bibfnamefont [1]{#1}%
\providecommand \citenamefont [1]{#1}%
\providecommand \href@noop [0]{\@secondoftwo}%
\providecommand \href [0]{\begingroup \@sanitize@url \@href}%
\providecommand \@href[1]{\@@startlink{#1}\@@href}%
\providecommand \@@href[1]{\endgroup#1\@@endlink}%
\providecommand \@sanitize@url [0]{\catcode `\\12\catcode `\$12\catcode
  `\&12\catcode `\#12\catcode `\^12\catcode `\_12\catcode `\%12\relax}%
\providecommand \@@startlink[1]{}%
\providecommand \@@endlink[0]{}%
\providecommand \url  [0]{\begingroup\@sanitize@url \@url }%
\providecommand \@url [1]{\endgroup\@href {#1}{\urlprefix }}%
\providecommand \urlprefix  [0]{URL }%
\providecommand \Eprint [0]{\href }%
\providecommand \doibase [0]{https://doi.org/}%
\providecommand \selectlanguage [0]{\@gobble}%
\providecommand \bibinfo  [0]{\@secondoftwo}%
\providecommand \bibfield  [0]{\@secondoftwo}%
\providecommand \translation [1]{[#1]}%
\providecommand \BibitemOpen [0]{}%
\providecommand \bibitemStop [0]{}%
\providecommand \bibitemNoStop [0]{.\EOS\space}%
\providecommand \EOS [0]{\spacefactor3000\relax}%
\providecommand \BibitemShut  [1]{\csname bibitem#1\endcsname}%
\let\auto@bib@innerbib\@empty
\bibitem [{\citenamefont {Kubo}\ \emph {et~al.}(1991)\citenamefont {Kubo},
  \citenamefont {Toda},\ and\ \citenamefont {Hashitsume}}]{Kubo1991}%
  \BibitemOpen
  \bibfield  {author} {\bibinfo {author} {\bibfnamefont {R.}~\bibnamefont
  {Kubo}}, \bibinfo {author} {\bibfnamefont {M.}~\bibnamefont {Toda}},\ and\
  \bibinfo {author} {\bibfnamefont {N.}~\bibnamefont {Hashitsume}},\ }\href
  {https://doi.org/10.1007/978-3-642-58244-8} {\emph {\bibinfo {title}
  {{Statistical Physics II}}}},\ \bibinfo {series} {Springer Series in
  Solid-State Sciences}, Vol.~\bibinfo {volume} {31}\ (\bibinfo  {publisher}
  {Springer Berlin Heidelberg},\ \bibinfo {address} {Berlin, Heidelberg},\
  \bibinfo {year} {1991})\BibitemShut {NoStop}%
\bibitem [{\citenamefont {Callen}(1985)}]{Callen1985}%
  \BibitemOpen
  \bibfield  {author} {\bibinfo {author} {\bibfnamefont {H.~B.}\ \bibnamefont
  {Callen}},\ }\href@noop {} {\emph {\bibinfo {title} {{Thermodynamics and an
  Introduction to Thermostatistics}}}},\ \bibinfo {edition} {2nd}\ ed.\
  (\bibinfo  {publisher} {John Wiley {\&} Sons, Inc.},\ \bibinfo {year}
  {1985})\ p.\ \bibinfo {pages} {493}\BibitemShut {NoStop}%
\bibitem [{\citenamefont {Gross}(2001)}]{Gross2001}%
  \BibitemOpen
  \bibfield  {author} {\bibinfo {author} {\bibfnamefont {D.~H.~E.}\
  \bibnamefont {Gross}},\ }\href {https://doi.org/10.1142/4340} {\emph
  {\bibinfo {title} {{Microcanonical Thermodynamics}}}}\ (\bibinfo  {publisher}
  {World Scientific, Singapore},\ \bibinfo {year} {2001})\BibitemShut {NoStop}%
\bibitem [{\citenamefont {Yoneta}\ and\ \citenamefont
  {Shimizu}(2019)}]{Yoneta2019}%
  \BibitemOpen
  \bibfield  {author} {\bibinfo {author} {\bibfnamefont {Y.}~\bibnamefont
  {Yoneta}}\ and\ \bibinfo {author} {\bibfnamefont {A.}~\bibnamefont
  {Shimizu}},\ }\bibfield  {title} {\bibinfo {title} {{Squeezed ensemble for
  systems with first-order phase transitions}},\ }\href
  {https://doi.org/10.1103/PhysRevB.99.144105} {\bibfield  {journal} {\bibinfo
  {journal} {Phys. Rev. B}\ }\textbf {\bibinfo {volume} {99}},\ \bibinfo
  {pages} {144105} (\bibinfo {year} {2019})}\BibitemShut {NoStop}%
\bibitem [{\citenamefont {Goldenfeld}(1992)}]{Goldenfeld1992}%
  \BibitemOpen
  \bibfield  {author} {\bibinfo {author} {\bibfnamefont {N.}~\bibnamefont
  {Goldenfeld}},\ }\href@noop {} {\emph {\bibinfo {title} {{Lectures on phase
  transitions and the renormalization group}}}}\ (\bibinfo  {publisher}
  {Addison-Wesley, Reading, Massachusetts},\ \bibinfo {year}
  {1992})\BibitemShut {NoStop}%
\bibitem [{\citenamefont {Lebowitz}(1999)}]{Lebowitz1999}%
  \BibitemOpen
  \bibfield  {author} {\bibinfo {author} {\bibfnamefont {J.~L.}\ \bibnamefont
  {Lebowitz}},\ }\bibfield  {title} {\bibinfo {title} {{Statistical mechanics:
  A selective review of two central issues}},\ }\href
  {https://doi.org/10.1103/revmodphys.71.s346} {\bibfield  {journal} {\bibinfo
  {journal} {Rev. Mod. Phys.}\ }\textbf {\bibinfo {volume} {71}},\ \bibinfo
  {pages} {S346} (\bibinfo {year} {1999})}\BibitemShut {NoStop}%
\bibitem [{\citenamefont {Yoneta}\ and\ \citenamefont
  {Shimizu}(2021)}]{Yoneta2021}%
  \BibitemOpen
  \bibfield  {author} {\bibinfo {author} {\bibfnamefont {Y.}~\bibnamefont
  {Yoneta}}\ and\ \bibinfo {author} {\bibfnamefont {A.}~\bibnamefont
  {Shimizu}},\ }\bibfield  {title} {\bibinfo {title} {Statistical ensembles for
  phase coexistence states specified by noncommutative additive observables},\
  }\href@noop {} {\  (\bibinfo {year} {2021})},\ \Eprint
  {https://arxiv.org/abs/2111.10532} {arXiv:2111.10532 [cond-mat.stat-mech]}
  \BibitemShut {NoStop}%
\bibitem [{\citenamefont {Kindermann}\ and\ \citenamefont
  {Snell}(1980)}]{Kindermann1980}%
  \BibitemOpen
  \bibfield  {author} {\bibinfo {author} {\bibfnamefont {R.}~\bibnamefont
  {Kindermann}}\ and\ \bibinfo {author} {\bibfnamefont {J.~L.}\ \bibnamefont
  {Snell}},\ }\href {https://doi.org/10.1090/conm/001} {\emph {\bibinfo {title}
  {{Markov Random Fields and Their Applications}}}},\ \bibinfo {series}
  {Contemporary Mathematics}, Vol.~\bibinfo {volume} {1}\ (\bibinfo
  {publisher} {American Mathematical Society},\ \bibinfo {address} {Providence,
  Rhode Island},\ \bibinfo {year} {1980})\BibitemShut {NoStop}%
\bibitem [{\citenamefont {Wilming}\ \emph {et~al.}(2017)\citenamefont
  {Wilming}, \citenamefont {Kastoryano}, \citenamefont {Werner},\ and\
  \citenamefont {Eisert}}]{Wilming2017}%
  \BibitemOpen
  \bibfield  {author} {\bibinfo {author} {\bibfnamefont {H.}~\bibnamefont
  {Wilming}}, \bibinfo {author} {\bibfnamefont {M.~J.}\ \bibnamefont
  {Kastoryano}}, \bibinfo {author} {\bibfnamefont {A.~H.}\ \bibnamefont
  {Werner}},\ and\ \bibinfo {author} {\bibfnamefont {J.}~\bibnamefont
  {Eisert}},\ }\bibfield  {title} {\bibinfo {title} {Emergence of spontaneous
  symmetry breaking in dissipative lattice systems},\ }\href
  {https://doi.org/10.1063/1.4978328} {\bibfield  {journal} {\bibinfo
  {journal} {J. Math. Phys.}\ }\textbf {\bibinfo {volume} {58}},\ \bibinfo
  {pages} {033302} (\bibinfo {year} {2017})}\BibitemShut {NoStop}%
\bibitem [{Note1()}]{Note1}%
  \BibitemOpen
  \bibinfo {note} {$\protect \hat {H}_N=\protect \hat {X}_{ \protect
  \ensuremath {{N,0}}}$ is the Hamiltonian of the completely isolated system
  and does not include the interactions with external fields.}\BibitemShut
  {Stop}%
\bibitem [{\citenamefont {Schulman}(1980)}]{Schulman1980}%
  \BibitemOpen
  \bibfield  {author} {\bibinfo {author} {\bibfnamefont {L.~S.}\ \bibnamefont
  {Schulman}},\ }\bibfield  {title} {\bibinfo {title} {{Magnetisation
  probabilities and metastability in the Ising model}},\ }\href
  {https://doi.org/10.1088/0305-4470/13/1/025} {\bibfield  {journal} {\bibinfo
  {journal} {J. Phys. A. Math. Gen.}\ }\textbf {\bibinfo {volume} {13}},\
  \bibinfo {pages} {237} (\bibinfo {year} {1980})}\BibitemShut {NoStop}%
\bibitem [{\citenamefont {Binder}(1981)}]{Binder1981}%
  \BibitemOpen
  \bibfield  {author} {\bibinfo {author} {\bibfnamefont {K.}~\bibnamefont
  {Binder}},\ }\bibfield  {title} {\bibinfo {title} {{Finite size scaling
  analysis of ising model block distribution functions}},\ }\href
  {https://doi.org/10.1007/BF01293604} {\bibfield  {journal} {\bibinfo
  {journal} {Zeitschrift f{\"{u}}r Phys. B Condens. Matter}\ }\textbf {\bibinfo
  {volume} {43}},\ \bibinfo {pages} {119} (\bibinfo {year} {1981})}\BibitemShut
  {NoStop}%
\bibitem [{\citenamefont {Challa}\ \emph {et~al.}(1986)\citenamefont {Challa},
  \citenamefont {Landau},\ and\ \citenamefont {Binder}}]{Challa1986}%
  \BibitemOpen
  \bibfield  {author} {\bibinfo {author} {\bibfnamefont {M.~S.~S.}\
  \bibnamefont {Challa}}, \bibinfo {author} {\bibfnamefont {D.~P.}\
  \bibnamefont {Landau}},\ and\ \bibinfo {author} {\bibfnamefont
  {K.}~\bibnamefont {Binder}},\ }\bibfield  {title} {\bibinfo {title}
  {{Finite-size effects at temperature-driven first-order transitions}},\
  }\href {https://doi.org/10.1103/PhysRevB.34.1841} {\bibfield  {journal}
  {\bibinfo  {journal} {Phys. Rev. B}\ }\textbf {\bibinfo {volume} {34}},\
  \bibinfo {pages} {1841} (\bibinfo {year} {1986})}\BibitemShut {NoStop}%
\bibitem [{\citenamefont {Vollmayr}\ \emph {et~al.}(1993)\citenamefont
  {Vollmayr}, \citenamefont {Reger}, \citenamefont {Scheucher},\ and\
  \citenamefont {Binder}}]{Vollmayr1993}%
  \BibitemOpen
  \bibfield  {author} {\bibinfo {author} {\bibfnamefont {K.}~\bibnamefont
  {Vollmayr}}, \bibinfo {author} {\bibfnamefont {J.~D.}\ \bibnamefont {Reger}},
  \bibinfo {author} {\bibfnamefont {M.}~\bibnamefont {Scheucher}},\ and\
  \bibinfo {author} {\bibfnamefont {K.}~\bibnamefont {Binder}},\ }\bibfield
  {title} {\bibinfo {title} {{Finite size effects at thermally-driven first
  order phase transitions: A phenomenological theory of the order parameter
  distribution}},\ }\href {https://doi.org/10.1007/BF01316713} {\bibfield
  {journal} {\bibinfo  {journal} {Zeitschrift f{\"{u}}r Phys. B Condens.
  Matter}\ }\textbf {\bibinfo {volume} {91}},\ \bibinfo {pages} {113} (\bibinfo
  {year} {1993})}\BibitemShut {NoStop}%
\bibitem [{Note2()}]{Note2}%
  \BibitemOpen
  \bibinfo {note} {It is sometimes possible to obtain the desired state at the
  first-order phase transition point using the {canonical ensemble} by imposing
  clever boundary conditions \cite {Dobrushin1973,vanBeijeren1975,Landau2014}.
  In this case, however, one would artificially add boundary terms to additive
  observables $\protect \hat {\protect \bm {X}}_N$ and construct the {canonical
  ensemble} from a set of these observables $\protect \hat {\protect \bm
  {X}}_N'$. On the other hand, since the generator of the time evolution is
  determined purely from quantum mechanics, independent of such artificial
  manipulations in statistical mechanics, a set of observables appearing in the
  definition of $\protect \hat {\Gamma }_N$ are the same as the original set of
  additive observables $\protect \hat {\protect \bm {X}}_N (\protect \neq
  \protect \hat {\protect \bm {X}}_N')$. As a result, Eq.~\protect \eqref
  {eq:Nc-Gamma} no longer holds. Therefore, the canonical Gibbs state $\protect
  \hat {\rho }_N^\protect \mathrm {c}$ is not strictly time
  invariant.}\BibitemShut {Stop}%
\bibitem [{\citenamefont {Elliott}\ and\ \citenamefont
  {Wood}(1971)}]{Elliott1971}%
  \BibitemOpen
  \bibfield  {author} {\bibinfo {author} {\bibfnamefont {R.~J.}\ \bibnamefont
  {Elliott}}\ and\ \bibinfo {author} {\bibfnamefont {C.}~\bibnamefont {Wood}},\
  }\bibfield  {title} {\bibinfo {title} {{The Ising model with a transverse
  field. I. High temperature expansion}},\ }\href
  {https://doi.org/10.1088/0022-3719/4/15/023} {\bibfield  {journal} {\bibinfo
  {journal} {J. Phys. C Solid State Phys.}\ }\textbf {\bibinfo {volume} {4}},\
  \bibinfo {pages} {2359} (\bibinfo {year} {1971})}\BibitemShut {NoStop}%
\bibitem [{\citenamefont {Pfeuty}\ and\ \citenamefont
  {Elliott}(1971)}]{Pfeuty1971}%
  \BibitemOpen
  \bibfield  {author} {\bibinfo {author} {\bibfnamefont {P.}~\bibnamefont
  {Pfeuty}}\ and\ \bibinfo {author} {\bibfnamefont {R.~J.}\ \bibnamefont
  {Elliott}},\ }\bibfield  {title} {\bibinfo {title} {{The Ising model with a
  transverse field. II. Ground state properties}},\ }\href
  {https://doi.org/10.1088/0022-3719/4/15/024} {\bibfield  {journal} {\bibinfo
  {journal} {J. Phys. C Solid State Phys.}\ }\textbf {\bibinfo {volume} {4}},\
  \bibinfo {pages} {2370} (\bibinfo {year} {1971})}\BibitemShut {NoStop}%
\bibitem [{\citenamefont {Nagai}\ \emph {et~al.}(1987)\citenamefont {Nagai},
  \citenamefont {Yamada}, \citenamefont {Nishino},\ and\ \citenamefont
  {Miyatake}}]{Nagai1987}%
  \BibitemOpen
  \bibfield  {author} {\bibinfo {author} {\bibfnamefont {O.}~\bibnamefont
  {Nagai}}, \bibinfo {author} {\bibfnamefont {Y.}~\bibnamefont {Yamada}},
  \bibinfo {author} {\bibfnamefont {K.}~\bibnamefont {Nishino}},\ and\ \bibinfo
  {author} {\bibfnamefont {Y.}~\bibnamefont {Miyatake}},\ }\bibfield  {title}
  {\bibinfo {title} {{Monte Carlo studies of Ising ferromagnets and the Villain
  model in transverse fields}},\ }\href
  {https://doi.org/10.1103/PhysRevB.35.3425} {\bibfield  {journal} {\bibinfo
  {journal} {Phys. Rev. B}\ }\textbf {\bibinfo {volume} {35}},\ \bibinfo
  {pages} {3425} (\bibinfo {year} {1987})}\BibitemShut {NoStop}%
\bibitem [{\citenamefont {Rieger}\ and\ \citenamefont
  {Kawashima}(1999)}]{Rieger1999}%
  \BibitemOpen
  \bibfield  {author} {\bibinfo {author} {\bibfnamefont {H.}~\bibnamefont
  {Rieger}}\ and\ \bibinfo {author} {\bibfnamefont {N.}~\bibnamefont
  {Kawashima}},\ }\bibfield  {title} {\bibinfo {title} {{Application of a
  continuous time cluster algorithm to the two-dimensional random quantum Ising
  ferromagnet}},\ }\href {https://doi.org/10.1007/s100510050761} {\bibfield
  {journal} {\bibinfo  {journal} {Eur. Phys. J. B}\ }\textbf {\bibinfo {volume}
  {9}},\ \bibinfo {pages} {233} (\bibinfo {year} {1999})},\ \Eprint
  {https://arxiv.org/abs/9802104} {arXiv:9802104 [cond-mat]} \BibitemShut
  {NoStop}%
\bibitem [{\citenamefont {Bl{\"{o}}te}\ and\ \citenamefont
  {Deng}(2002)}]{Blote2002}%
  \BibitemOpen
  \bibfield  {author} {\bibinfo {author} {\bibfnamefont {H.~W.}\ \bibnamefont
  {Bl{\"{o}}te}}\ and\ \bibinfo {author} {\bibfnamefont {Y.}~\bibnamefont
  {Deng}},\ }\bibfield  {title} {\bibinfo {title} {{Cluster Monte Carlo
  simulation of the transverse Ising model}},\ }\href
  {https://doi.org/10.1103/PhysRevE.66.066110} {\bibfield  {journal} {\bibinfo
  {journal} {Phys. Rev. E}\ }\textbf {\bibinfo {volume} {66}},\ \bibinfo
  {pages} {8} (\bibinfo {year} {2002})}\BibitemShut {NoStop}%
\bibitem [{\citenamefont {Nakamura}\ and\ \citenamefont
  {Ito}(2003)}]{Nakamura2003}%
  \BibitemOpen
  \bibfield  {author} {\bibinfo {author} {\bibfnamefont {T.}~\bibnamefont
  {Nakamura}}\ and\ \bibinfo {author} {\bibfnamefont {Y.}~\bibnamefont {Ito}},\
  }\bibfield  {title} {\bibinfo {title} {{A quantum Monte Carlo algorithm
  realizing an intrinsic relaxation}},\ }\href
  {https://doi.org/10.1143/JPSJ.72.2405} {\bibfield  {journal} {\bibinfo
  {journal} {J. Phys. Soc. Japan}\ }\textbf {\bibinfo {volume} {72}},\ \bibinfo
  {pages} {2405} (\bibinfo {year} {2003})}\BibitemShut {NoStop}%
\bibitem [{\citenamefont {{du Croo de Jongh}}\ and\ \citenamefont {van
  Leeuwen}(1998)}]{Jongh1998}%
  \BibitemOpen
  \bibfield  {author} {\bibinfo {author} {\bibfnamefont {M.}~\bibnamefont {{du
  Croo de Jongh}}}\ and\ \bibinfo {author} {\bibfnamefont {J.}~\bibnamefont
  {van Leeuwen}},\ }\bibfield  {title} {\bibinfo {title} {{Critical behavior of
  the two-dimensional Ising model in a transverse field: A density-matrix
  renormalization calculation}},\ }\href
  {https://doi.org/10.1103/PhysRevB.57.8494} {\bibfield  {journal} {\bibinfo
  {journal} {Phys. Rev. B}\ }\textbf {\bibinfo {volume} {57}},\ \bibinfo
  {pages} {8494} (\bibinfo {year} {1998})}\BibitemShut {NoStop}%
\bibitem [{\citenamefont {Bacon}\ and\ \citenamefont
  {Pease}(1955)}]{Bacon1955}%
  \BibitemOpen
  \bibfield  {author} {\bibinfo {author} {\bibfnamefont {G.~E.}\ \bibnamefont
  {Bacon}}\ and\ \bibinfo {author} {\bibfnamefont {R.~S.}\ \bibnamefont
  {Pease}},\ }\bibfield  {title} {\bibinfo {title} {A neutron-diffraction study
  of the ferroelectric transition of potassium dihydrogen phosphate},\ }\href
  {https://doi.org/10.1098/rspa.1955.0139} {\bibfield  {journal} {\bibinfo
  {journal} {Proc. R. Soc. Lond. A}\ }\textbf {\bibinfo {volume} {230}},\
  \bibinfo {pages} {359^^e2^^80^^93381} (\bibinfo {year} {1955})}\BibitemShut
  {NoStop}%
\bibitem [{\citenamefont {de~Gennes}(1963)}]{DeGennes1963}%
  \BibitemOpen
  \bibfield  {author} {\bibinfo {author} {\bibfnamefont {P.~G.}\ \bibnamefont
  {de~Gennes}},\ }\bibfield  {title} {\bibinfo {title} {{Collective motions of
  hydrogen bonds}},\ }\href {https://doi.org/10.1016/0038-1098(63)90212-6}
  {\bibfield  {journal} {\bibinfo  {journal} {Solid State Commun.}\ }\textbf
  {\bibinfo {volume} {1}},\ \bibinfo {pages} {132} (\bibinfo {year}
  {1963})}\BibitemShut {NoStop}%
\bibitem [{\citenamefont {Serre}(1977)}]{Serre1977}%
  \BibitemOpen
  \bibfield  {author} {\bibinfo {author} {\bibfnamefont {J.-P.}\ \bibnamefont
  {Serre}},\ }\href {https://doi.org/10.1007/978-1-4684-9458-7} {\emph
  {\bibinfo {title} {Linear Representations of Finite Groups}}},\ \bibinfo
  {series} {Graduate Texts in Mathematics}, Vol.~\bibinfo {volume} {42}\
  (\bibinfo  {publisher} {Springer},\ \bibinfo {address} {New York, NY},\
  \bibinfo {year} {1977})\BibitemShut {NoStop}%
\bibitem [{Note3()}]{Note3}%
  \BibitemOpen
  \bibinfo {note} {Conversely, Eq.~\protect \eqref {eq:accidental_degeneracy}
  implies condition~\ref {cond:accidental_degeneracy}. Suppose Eq.~\protect
  \eqref {eq:accidental_degeneracy} holds. Then for any eigenstates $\mathinner
  {|{E_1}\rangle }$ and $\mathinner {|{E_2}\rangle }$ which belong to the same
  eigenvalue of $\protect \hat {\Gamma }_N$, we can expand as \begin {align}
  \mathinner {|{E_s}\rangle } &= \DOTSB \sum@ \slimits@ _a c_s^a \mathinner
  {|{E_{\mu \nu }^{ia}}\rangle } = \mathinner {|{i}\rangle }_{\mu \nu } \otimes
  \left (\DOTSB \sum@ \slimits@ _a c_s^a \mathinner {|{a}\rangle }_\nu \right )
  = \mathinner {|{i}\rangle }_{\mu \nu } \otimes \mathinner {|{\phi _s}\rangle
  }_\nu . \end {align} On the other hand, since $(\protect \mathcal {H}_\nu
  ,\sigma _\nu )$ is an irreducible representation of $F$, there exists
  $\protect \hat {V} \in F$ such that $\mathinner {\langle {\phi _1|\sigma _\nu
  (\protect \hat {V})|\phi _2}\rangle } \protect \neq 0$. Therefore, we have
  condition~\ref {cond:accidental_degeneracy}.}\BibitemShut {Stop}%
\bibitem [{\citenamefont {Lieb}\ and\ \citenamefont
  {Robinson}(1972)}]{Lieb1972}%
  \BibitemOpen
  \bibfield  {author} {\bibinfo {author} {\bibfnamefont {E.~H.}\ \bibnamefont
  {Lieb}}\ and\ \bibinfo {author} {\bibfnamefont {D.~W.}\ \bibnamefont
  {Robinson}},\ }\bibfield  {title} {\bibinfo {title} {{The finite group
  velocity of quantum spin systems}},\ }\href
  {https://doi.org/10.1007/BF01645779} {\bibfield  {journal} {\bibinfo
  {journal} {Commun. Math. Phys.}\ }\textbf {\bibinfo {volume} {28}},\ \bibinfo
  {pages} {251} (\bibinfo {year} {1972})}\BibitemShut {NoStop}%
\bibitem [{\citenamefont {Nachtergaele}\ and\ \citenamefont
  {Sims}(2006)}]{Nachtergaele2006}%
  \BibitemOpen
  \bibfield  {author} {\bibinfo {author} {\bibfnamefont {B.}~\bibnamefont
  {Nachtergaele}}\ and\ \bibinfo {author} {\bibfnamefont {R.}~\bibnamefont
  {Sims}},\ }\bibfield  {title} {\bibinfo {title} {{Lieb-Robinson bounds and
  the exponential clustering theorem}},\ }\href
  {https://doi.org/10.1007/s00220-006-1556-1} {\bibfield  {journal} {\bibinfo
  {journal} {Commun. Math. Phys.}\ }\textbf {\bibinfo {volume} {265}},\
  \bibinfo {pages} {119} (\bibinfo {year} {2006})}\BibitemShut {NoStop}%
\bibitem [{\citenamefont {Nachtergaele}\ and\ \citenamefont
  {Sims}(2010)}]{Nachtergaele2010}%
  \BibitemOpen
  \bibfield  {author} {\bibinfo {author} {\bibfnamefont {B.}~\bibnamefont
  {Nachtergaele}}\ and\ \bibinfo {author} {\bibfnamefont {R.}~\bibnamefont
  {Sims}},\ }\bibfield  {title} {\bibinfo {title} {{Lieb-Robinson bounds in
  quantum many-body physics}},\ }\href {https://doi.org/10.1090/conm/529/10429}
  {\bibfield  {journal} {\bibinfo  {journal} {Contemp. Math.}\ }\textbf
  {\bibinfo {volume} {529}},\ \bibinfo {pages} {141} (\bibinfo {year}
  {2010})}\BibitemShut {NoStop}%
\bibitem [{\citenamefont {Dobrushin}(1973)}]{Dobrushin1973}%
  \BibitemOpen
  \bibfield  {author} {\bibinfo {author} {\bibfnamefont {R.~L.}\ \bibnamefont
  {Dobrushin}},\ }\bibfield  {title} {\bibinfo {title} {{Gibbs State Describing
  Coexistence of Phases for a Three-Dimensional Ising Model}},\ }\href
  {https://doi.org/10.1137/1117073} {\bibfield  {journal} {\bibinfo  {journal}
  {Theory Probab. Its Appl.}\ }\textbf {\bibinfo {volume} {17}},\ \bibinfo
  {pages} {582} (\bibinfo {year} {1973})}\BibitemShut {NoStop}%
\bibitem [{\citenamefont {van Beijeren}(1975)}]{vanBeijeren1975}%
  \BibitemOpen
  \bibfield  {author} {\bibinfo {author} {\bibfnamefont {H.}~\bibnamefont {van
  Beijeren}},\ }\bibfield  {title} {\bibinfo {title} {{Interface sharpness in
  the Ising system}},\ }\href {https://doi.org/10.1007/BF01614092} {\bibfield
  {journal} {\bibinfo  {journal} {Commun. Math. Phys.}\ }\textbf {\bibinfo
  {volume} {40}},\ \bibinfo {pages} {1} (\bibinfo {year} {1975})}\BibitemShut
  {NoStop}%
\bibitem [{\citenamefont {Landau}\ and\ \citenamefont
  {Binder}(2014)}]{Landau2014}%
  \BibitemOpen
  \bibfield  {author} {\bibinfo {author} {\bibfnamefont {D.~P.}\ \bibnamefont
  {Landau}}\ and\ \bibinfo {author} {\bibfnamefont {K.}~\bibnamefont
  {Binder}},\ }\href {https://doi.org/10.1017/cbo9781139696463} {\emph
  {\bibinfo {title} {A Guid. to Monte Carlo Simulations Stat. Phys.}}}\
  (\bibinfo  {publisher} {Cambridge University Press},\ \bibinfo {year}
  {2014})\BibitemShut {NoStop}%
\end{thebibliography}%
\end{document}